\documentclass{article}

\usepackage[backend=bibtex]{biblatex}
\usepackage[sts,preprint]{imsart1178}
\usepackage{amsthm}
\usepackage{amssymb}
\usepackage{amsmath}
\usepackage{graphicx}

\arxiv{stat.ME/1405.7395}

\newtheorem{theorem}{Theorem}
\newtheorem{lemma}{Lemma}

\begin{document}

\begin{frontmatter}

\title{Partial exchangeability of the prior via shuffling}

\begin{aug}
\author{E.W. van Zwet} \ead[label=e1]{vanzwet@lumc.nl}
\affiliation{Leiden University Medical Center}
\address{P.O. Box 9600, 2300 RC Leiden, The Netherlands, \printead{e1}}
\end{aug}

\begin{abstract}
In inference problems involving a multi-dimensional parameter $\theta$, it is often natural to consider decision rules that have a risk which is invariant under some group $G$ of permutations of $\theta$. We show that this implies that the Bayes risk of the rule is {\em as if} the prior distribution of the parameter is partially exchangeable with respect to $G$. We provide a symmetrization technique for incorporating partial exchangeability of $\theta$ into a statistical model, without assuming any other prior information. We refer to this technique as {\em shuffling}. Shuffling can be viewed as an instance of empirical Bayes, where we estimate the (unordered) multiset of parameter values $\{\theta_1,\theta_2,\dots,\theta_p\}$ while using a uniform prior on $G$ for their ordering. Estimation of the multiset is a missing data problem which can be tackled with a stochastic EM algorithm. We show that in the special case of estimating the mean-value parameter in a regular exponential family model, shuffling leads to an estimator that is a weighted average of permuted versions of the usual maximum likelihood estimator. This is a novel form of shrinkage.
\end{abstract}


\end{frontmatter}

\section{Introduction}
We start with the following compound (or simultaneous) estimation problem. Suppose we observe  $X$ with distribution $f(x | \theta)$, where $\theta=(\theta_1,\theta_2,\dots,\theta_p) \in \Theta^p$ is an unknown parameter vector.  The goal is to estimate $\theta$ under squared error loss. The maximum likelihood estimator (MLE)
\begin{equation}\label{mle}
\hat \theta = \max_{\theta \in  \Theta^p} f(x | \theta)
\end{equation}
is often asymptotically optimal, but its finite sample performance can still be disappointing. We briefly discuss two classical examples.

\medskip \noindent
{\bf Example 1} (James--Stein)  Suppose we observe a sample $X$ from the multivariate normal distribution with mean vector $\theta \in {\mathbb R}^p$ and covariance $I$. The obvious estimator $\hat \theta = X$ is maximum likelihood, minimum variance unbiased and meets a host of other optimality criteria. Therefore, it came as quite a surprise when Stein \cite{stein1956inadmissibility} proved that if $p \geq 3$, it is not admissible under squared error loss. James and Stein \cite {james1961estimation} proceeded to construct their well-known shrinkage estimator.

\medskip \noindent
{\bf Example 2} (species sampling) Suppose we have a sample $X=(X_1,X_2,\dots,X_p)$ from the multinomial distribution with parameters $n$ and $\theta \in  {\mathbb R}^p_{\geq 0}$ such that $\sum_{i=1}^p \theta_i = 1$. Again, the obvious estimator $\hat \theta = X/n$ is maximum likelihood, minimum variance unbiased et cetera. However, consider the following situation. Suppose $p$ is quite large, say $p=100$, and in $n=10$ trials we have not encountered the same outcome twice. Common sense suggests that there must be many outcomes with (small) positive probability that we have not seen yet. The MLE disagrees; it puts unit probability mass on the set of observed outcomes. This behaviour is problematic in all applications where there are many probabilities that are small in relation to the sample size. In such cases, a substantial part of the probability mass will be unobserved and misattributed to the observed outcomes for by the MLE. Interestingly, the total unobserved probability mass can easily be estimated by the Good--Turing estimator  \cite{good1953population} which is just the proportion of singleton observations.

\medskip \noindent
The empirical Bayes (EB) approach has been very successful in compound estimation problems. We refer to  Efron \cite{efron2003robbins} for a discussion. Consider the following hierarchical model.
\begin{enumerate}
\item Sample $\theta=(\theta_1,\theta_2,\dots,\theta_p)$  from a distribution $P \in \cal P$
\item Observe a sample $X$ from $f(x | \theta)$.
\end{enumerate}
In a sense this a frequentist model, because it relies on a fixed, unknown parameter (distribution) $P$. We use the data to find some estimate $\hat P$ of $P$ and then estimate $\theta$ by its conditional expectation $E_{\hat P}(\theta \mid X)$. If we do not wish to assume too much about $\theta$, we can take the set ${\cal P}$ to be very large. For instance, we could take $\theta$ to be an i.i.d.\ sample from an arbitrary distribution. This is the nonparametric version of EB, see \cite{laird1978nonparametric}. A recent study \cite{jiang2009general} shows good performance -- both in theory and in simulations -- in the multivariate normal model with unknown means when $p$ becomes large.

A fully Bayesian approach to the problem would be to add a level to the hierarchy by putting an (hyper)prior distribution on the set $\cal P$. Then $\theta$ could be estimated by the posterior mean $E(\theta \mid X)$. If we do not want to assume any external information about $\theta$, we should turn to objective Bayesian methods \cite{berger2006case}. In single parameter problems, Jeffrey's prior is a possible choice \cite{jeffreys1946invariant} and in multi-dimensional problems, the reference prior \cite{bernardo1979reference} could be used. However,  it seems that in multi-parameter problems, the prior must depend on the choice of a one-dimensional parameter of interest to get posterior distributions with optimal properties \cite{overall}. This is awkward if we are equally interested in each $\theta_i$.

\section{Shuffling}
Suppose that, from a Bayesian point of view, we are willing to assume partial exchangeability of the parameter vector $\theta$ with respect to some group $G$ of permutations of the index set $\{1,2,\dots,p\}$, c.f.\ \cite{aldous1981representations}. This means that the joint distribution of $\theta_1,\theta_2,\dots,\theta_p$ is invariant under all permutations in $G$,
$$(\theta_1,\theta_2,\dots,\theta_p) \overset{d}{=}(\theta_{\pi(1)},\theta_{\pi(2)},\dots,\theta_{\pi(p)}), \quad \text{for every\ }\pi \in G.$$
If $G$ is the group of {\em all} permutations of  $\{1,2,\dots,p\}$, then $\theta$ is called (fully) exchangeable. Now consider the following {\em shuffled} model
\begin{enumerate}
\item Fix an arbitrary $\psi=(\psi_1,\psi_2,\dots,\psi_p) \in \Theta^p$.
\item Let $\Pi$ be a completely random element of a group $G$ permutations of $\{1,2,\dots,p\}$ and set
$$\theta=\psi_\Pi=(\psi_{\Pi(1)},\psi_{\Pi(2)},\dots,\psi_{\Pi(p)}).$$
\item Observe a sample $X$ from $f(x | \theta)$.
\end{enumerate}
This is  a frequentist model, with an arbitray but  fixed parameter vector $\psi$. However, we also have a random vector $\theta$ which, by construction, is partially exchangeable with respect to $G$.  If $G$ is taken to consist only of the identity, then the shuffled model reduces to the usual frequentist model. Also, if $p=1$ then the shuffled model and the usal model coincide. Note that even if $f(x | \theta)=\prod_{i=1}^n f(x_i \mid \theta)$, then the $X_i$ are not necessarily independent in the shuffled model.

Now, $\psi$ is a fixed parameter and we can estimate it from the data $X$, for instance by using the method of maximum likelihood. That is,
\begin{equation}\label{lik}
\hat \psi = \underset{\psi \in \Theta^p}{\max}\ \frac{1}{|G|} \sum_{\pi \in G} f_{\psi_\pi}(X)
\end{equation}
where $|G|$ is the number of elements of $G$. Clearly, the likelihood is invariant under permutations of $\psi$ and hence $\psi$ can only be identified up to the permutations in $G$. The lack of identifiability will turn out not to be a problem, but if so desired it can be defined away by assuming that $\psi$ is ordered.

Noting that $\theta=\psi_\Pi$ is now random, we can estimate it by 
\begin{equation}
\hat \theta_{\rm shuffle} = E_{\hat \psi}(\theta \mid X)=E_{\hat \psi}(\hat \psi_\Pi \mid X).
\end{equation}
The expectation is invariant under permutations of $\hat \psi$ and so it does not matter that  $\hat \psi$ is only defined up to the permutations in $G$.  

The estimator (\ref{lik}) has been considered in the multinomial case by Orlitsky et al.\ \cite{orlitsky2003always}. They refer to $\hat \psi$ as the ``high profile estimator'' or ``pattern MLE''. At least in the multinomial case, it exists and is unique (up to permutation). It is also consistent in the sense that the ordered $\hat \psi$  converges to the ordered $\theta$ as the sample size $n$ tends to infinity, see \cite{orlitsky2005convergence}, \cite{anevski2013estimating}.

The shuffled estimator can be very different from the usual MLE. As an illustration \cite{orlitsky2004modeling}, consider $X$ having the binomial distribution with parameters $n=3$ and probability of success $p$ and suppose we observe $X=2$. The ordinary MLE is of course $\hat p=2/3$. 
In the shuffled model, we assume exchangeability of $p$ and $1-p$. An easy computation then shows that MLE of $\psi$ is $\hat \psi=(1/2,1/2)$, which implies $\hat p_{\rm shuffle}=1/2$. This may seem wrong, but note that it is not possible to have a more even distribution of failures and successes in 3 trials. In other words, observing 2 successes in 3 trials is in perfect agreement with $p=1/2$. Is there really any evidence at all to conclude that successes are twice as likely as failures?

\section{Decision theory}
Exchangeability of the parameter vector is a fairly mild assumption, but in certain applications it might not be appropriate. Also, from a frequentist perspective, $\theta$ is not random and so exchangeability does not apply. However, in this section we show that the shuffled model can be motivated by a decision theoretic argument that does not depend on whether the parameter is exchangeable, or not.

Recall the usual decision theoretic framework. If we observe $X$ then we take action $a(X) \in A$, which results in a loss $L(\theta,a(X))$ when the parameter is $\theta$. The function $a(\cdot)$ is called a decision rule. The expected loss, or risk, is
$$R(\theta,a)=E_\theta L(\theta,a(X)).$$
We call the risk {\em permutation invariant} with respect to $G$ if 
\begin{equation}
R(\theta,a)=R(\theta_\pi,a)\quad \text{for all}\ \pi \in G.
\end{equation}
For example, suppose $X$ has the multivariate normal distribution with mean $\theta \in {\mathbb R}^p$ and covariance $I$ and we estimate $\theta$ by $\hat \theta = X$. If $G$ is the group of all permutations of $\{1,2,\dots,p\}$  then the risk of $\hat \theta$ under squared error loss  is permutation invariant with respect to $G$. In fact, it is permutation invariant under any $\ell_p$ loss. 

Estimation of one-dimensional parameters that are permutation invariant functions of $\theta$ usually lead to permutation invariant risk. Examples of such parameters are $\sum_{i=1}^p (\theta_i - \bar \theta)^2$ and $\max(\theta_1,\theta_2,\dots,\theta_p)$.

Besides estimation, there are other decision problems with permutation invariant risk. We might be interested in ranking the $\theta_i$, for instance  when they refer to the performance of schools or hospitals. Or, if the $\theta_i$ refer to the effects of different genes on some phenotype, then we might be interested in testing all the hypotheses $H_0 : \theta_i=0$. In such situations, we will typically have permutation invariant risk. An example of a situation where the risk may be {\em not} permutation invariant, is estimation of $\theta_1$.

Suppose Nature chooses $\theta$ according to some distribution $\Lambda$ on the parameter space $\Theta^p$. Note that we are not assuming that the $\theta_i$ are exchangeable under $\Lambda$. In fact, $\Lambda$ could be point mass at some specific value which would correspond to the frequentist point of view. The Bayes risk is
\begin{equation}
{\cal R}(\Lambda,a)=\int_{\Theta^p}  R(\theta,a) \rm{d} \Lambda(\theta).
\end{equation}
Now define a distribution $\Lambda^*$ by symmetrizing $\Lambda$ over $G$. Set
\begin{equation}\label{symm}
\Lambda^*(S) =  \sum_{\pi \in G} \int_{\{\theta : \theta_\pi \in S\}}  \frac{1}{|G|} {\rm d} \Lambda(\theta)
\end{equation}
for all (measurable) $S \subseteq \Theta^p$.  We can obtain a sample from $\Lambda^*$ by drawing $\theta$ from $\Lambda$, selecting a permutation $\pi$ uniformly at random from $G$ and then permuting $\theta$ according to $\pi$. This is the construction we used in the shuffled model of the previous section. Evidently, $\Lambda^*$ is partially exchangeable with respect to $G$.  The following lemma characterizes partially exchangeable distributions as symmetrizations:

\begin{lemma}
A distribution on $\Theta^p$ is partially exchangeable with respect to $G$ if and only if it can be obtained by symmetrizing a distribution on $\Theta^p$ over $G$.
\end{lemma}
\begin{proof}
A symmetrized distribution is partially exchangeable by construction. For the converse, note that a partially exchangeable distribution is invariant under symmetrization. Thus any partially exchangeable distribution is a symmetrization of itself.
\end{proof}
The lemma confirms that by using symmetrization over $G$, the shuffled model is indeed assuming partial exchangeability of the prior with respect to $G$ and nothing else. We also have the following.

\begin{theorem}
Consider the definitions and notation of this section. If one of the following two conditions holds:
\begin{enumerate}
\item The prior $\Lambda$ for $\theta$ is partially exchangeable with respect to $G$
\item The risk of decision rule $a(\cdot)$ is permutation invariant with respect to $G$
\end{enumerate}
Then
$${\cal R}(\Lambda,a) = {\cal R}(\Lambda^*,a).$$
\end{theorem}

\begin{proof}
As we noted in the preceding lemma, a partially exchangeable distribution is invariant under symmetrization. So, if the first condition holds, then $\Lambda$ and $\Lambda^*$ coincide and hence ${\cal R}(\Lambda,a) = {\cal R}(\Lambda^*,a)$. If the second condition holds, then
\begin{equation}
{\cal R}(\Lambda,a) =\int_{\Theta^p} R(\theta,a) {\rm d} \Lambda(\theta) =\int_{\Theta^p} \sum_{\pi \in G} R(\theta_\pi,a) \frac{1}{|G|} {\rm d} \Lambda(\theta) = {\cal R}(\Lambda^*,a).
\end{equation}
\end{proof}
The theorem says that if we are considering a decision rule with permutation invariant risk, then the Bayes risk is {\em as if} Nature has symmetrized the prior. So, if for some choice of loss function, we restrict ourselves to decision rules with permutation invariant risk, then we are entitled to assume exchangeability of the parameter. Using the shuffled model is a way to do that, without assuming any other prior information.

Exchangeability of the parameter is justified whenever the risk is permutation invariant. It is particularly interesting that this also holds for a frequentist who would not even consider $\theta$ to be random!

Let us reconsider the James--Stein phenomenon. The inadmissibility of the MLE in the normal location model when $p \geq 3$ is quite mysterious and it is sometimes referred to as Stein's paradox. As the $X_i$ are all independent, it seems that there is no information in $X_2,X_3,\dots,X_p$ about $\theta_1$, and yet Stein \cite{stein1956inadmissibility} proved that there is. Where is the information coming from? Various explanations have been suggested but none entirely satisfactory \cite{stigler19901988}. Efron and Morris  \cite{efron1973stein} noted that the James--Stein estimator can be motivated from an empirical Bayes perspective, but Stein proved his result without relying on any assumptions about the distribution of the $\theta_i$. He did not even assume the $\theta_i$ {\em have} a distribution. However, note that both the MLE and the James--Stein estimator have a risk (under total squared error loss) that is invariant under permutation. While the James--Stein estimator exploits the implied exchangeability of the parameter, the MLE does not.

\section{Shrinkage in exponential families}
It is interesting and useful (for computation, see next section) to view estimation of the parameter $\psi$ of the shuffled model as a missing data problem. The complete data are the pair $(X,\Pi)$ while we observe only $X$. Maximum likelihood estimation with missing data is relatively easy in the case of exponential families. Therefore, we will now turn to this special case. 
So, suppose that the distribution of the data follows a regular exponential family with a density of the form
\begin{equation}
h(x) \exp\{\eta \cdot t(x)\}/ a(\eta)
\end{equation}
where $a(\eta) = \int h(x) \exp\{\eta \cdot t(x)\} dx$. A familiar example is the multivariate normal distribution with mean $\mu=(\mu_1,\mu_2,\dots,\mu_{p-1})$ and covariance $\sigma^2 I$. The canonical parameter is
$$\eta=\left(\frac{\mu_1}{\sigma^2},\frac{\mu_2}{\sigma^2},\dots,\frac{\mu_{p-1}}{\sigma^2},\frac{-1}{2\sigma^2}\right)$$
and the sufficient statistic is $t(x)=(x_1,x_2,\dots,x_{p-1},\sum_{i=1}^{p-1} x_i^2).$
The MLE of the (canonical) parameter $\eta$ is the solution of the likelihood equation
\begin{equation}
t(x) - E_\eta (t(X)) = 0.
\end{equation}
Consider the mean-value mapping $\theta(\eta) =  E_\eta(t(X))$. Then the likelihood equation for the parameter $\theta$ is trivial: $t(x) - \theta=0$. Since the likelihood equation has at most one solution, the mapping $\theta(\eta)$ is one-to-one, hence invertible.  So, we can re-parameterize the exponential family in terms of $\theta$ and we write
\begin{equation}
f(x | \theta)=h(x) \exp\{\eta(\theta) \cdot t(x)\} / a(\eta(\theta)).
\end{equation}
In the normal example, the mean value parameter is
$$\theta = \left( \mu_1,\mu_2,\dots,\mu_{p-1},\sum_{i=1}^{p-1} (\sigma^2 + \mu_i^2) \right)$$
and the mapping from the mean value parameter to the canonical parameter is given by
$$\eta(\theta)=\left( \frac{(p-1)\theta_1}{\theta_p - \sum_{i=1}^{p-1} \theta_i^2},\dots, \frac{(p-1)\theta_{p-1}}{\theta_p - \sum_{i=1}^{p-1} \theta_i^2},\frac{-(p-1)}{2(\theta_p - \sum_{i=1}^{p-1} \theta_i^2)}\right).$$
Let us now turn to the shuffled model with parameter $\psi$. We sample a permutation $\Pi$ uniformly at random from a group $G$ and given $\Pi=\pi$, we sample $X$ from an exponential family with mean value parameter $\theta=\psi_\pi$. Thus, the pair $(X,\Pi)$ is distributed according to 
\begin{align}
f(x,\pi | \psi) &= \frac{1}{|G|} f(x | \psi_\pi) =\frac{1}{|G|} h(x) \exp\{\eta(\psi_\pi) \cdot t(x)\} / a(\eta(\psi_\pi)).
\end{align}
Suppose that for all permutations $\pi \in G$ we have
\begin{equation}\label{perm cond}
\eta(\psi_\pi) = \eta(\psi)_\pi.
\end{equation}
In the normal example, the group of all permutations of the first $p-1$ elements of the mean value parameter satisfies this condition. The group of {\em all} permutations of  the mean value parameter does not satisfy this condition. So, here the condition means that we must not swap $\sigma^2$ with any of the $\mu_i$. If (\ref{perm cond}) does hold for all $\pi \in G$, then
\begin{align}
f(x,\pi | \psi) &= \frac{1}{|G|} h(x) \exp\{\eta(\psi_\pi) \cdot t(x)\} / a(\eta(\psi_\pi)) \notag \\
&=\frac{1}{|G|} h(x) \exp\{\eta(\psi)_\pi \cdot t(x)\} /\int h(x) \exp\{\eta(\psi)_\pi \cdot t(x)\} dx \notag \\
&=\frac{1}{|G|} h(x) \exp\{\eta(\psi) \cdot t(x)_{\pi^{-1}}\} / \int h(x) \exp\{\eta(\psi) \cdot t(x)_{\pi^{-1}}\} dx.
\end{align}
We see that  the pair $(X,\Pi)$ is distributed according to an exponential family with sufficient statistic $t(X)_{\Pi^{-1}}$ and mean-value parameter $\psi$.  The complete data MLE of $\psi$ is the sufficient statistic $t(X)_{\Pi^{-1}}$.

Of course we observe only $X$. However, as we have assumed a {\em regular} exponential family, the likelihood equation for the observed data has a very simple form \cite{sundberg1974maximum}
\begin{equation}\label{score}
E_\psi(t(X)_{\Pi^{-1}} \mid X) - E_\psi(t(X)_{\Pi^{-1}})= E_\psi(t(X)_{\Pi^{-1}} \mid X) - \psi = 0.
\end{equation}
We find that the (an) observed data MLE $\hat \psi$ of $\psi$, is a weighted average of permutations of $t(X)$. Next, we estimate $\theta=\psi_\Pi$ by the conditional expectation
\begin{equation}\label{shuffled}
\hat \theta_{\rm shuffle} = E_{\hat \psi}(\theta \mid X) =  E_{\hat \psi}(\hat \psi_\Pi \mid X) = E_{\hat \psi}(E_{\hat \psi}(t(X)_{\Pi^{-1}} \mid X)_\Pi \mid X).
\end{equation}
So,  $\hat \theta_{\rm shuffle}$ is a also a weighted average of permutations of $t(X)$. Since $\hat \theta = t(X)$ is the MLE of $\theta$, we find that $\hat \theta_{\rm shuffle}$ is a weighted average of permations of the MLE. This is a shrinkage effect because the range of $\hat \theta_{\rm shuffle}$ is smaller than (or equal to) the range of $\hat \theta$. However, there is no single shrinkage target, and no single shrinkage factor. Instead, each individual estimate $\hat \theta_i$ is pushed towards a nearby area of the parameter space $\Theta$ where many of the other elements of $\hat \theta$ are.

\section{Computation}
Maximizing the likelihood (\ref{lik}) is not a trivial matter because even for moderate $p$ the sum involves very many terms.  In this section we discuss the stochastic approximation (SA) EM algorithm of \cite{delyon1999convergence}, \cite{kuhn2004coupling}. This algorithm was used in 
\cite{anevski2013estimating} to maximize the likelihood (\ref{lik}) in the multinomial case. We also describe how the algorithm can be extended to deliver $\hat \theta_{\rm shuffle}$ as well as the MLE  $\hat \psi$. We have used the algorithm to obtain the results presented in the next section, but faster algorithms would certainly be needed to deal with problems where $p$ is large.

We view estimation of $\psi$ as a missing data problem, where the complete data is the pair $(X,\Pi)$ and the observed data is $X$. This leads us to consider the familiar EM algorithm \cite{dempster1977maximum}. Starting from an initial estimate $\psi^0$, we iterate the following two steps.

\medskip \noindent
E-step: Compute $Q(\psi | \psi^k) = E_{\psi^k} ( \log f(X ,\Pi \mid \psi) \mid X=x)$

\noindent
M-step: Maximize $Q(\psi | \psi^k)$ as a function of $\psi$ to obtain $\psi^{k+1}$.

\medskip \noindent
If (\ref{score}) holds, then the E and M steps collapse into a single step
\begin{equation}\label{cond expec}
\psi^{k+1} = E_{\psi^k} ( t(X)_{\Pi^{-1}} \mid  X ).
\end{equation}
Unfortunately, the EM algorithm as described is not practical because the conditional expectation again involves a sum over very many terms. However, for given $\psi$, we can use the Metropolis--Hastings algorithm to obtain samples from the conditional distribution of $\Pi$ given $X$.

Start with some initial choice for $\pi \in G$, say the identity. Propose $\pi'$ by selecting uniformly at random distinct $i$ and $j$ in $\{1,2,\dots,p\}$ and setting $\pi'(i)=\pi(j)$ and $\pi'(j)=\pi(i)$. The logarithm of the acceptance ratio of this move is
\begin{equation}
H = \log f(x,\pi' | \psi) -  \log f(x,\pi | \psi).
\end{equation}
Now sample $Z$ from the standard exponential distribution and accept $\pi'$ if and only if $-Z \leq H$.  Recall that if $Z$ has the standard exponential distribution, then $\exp(-Z)$ has the standard uniform distribution. It follows that if we repeat this procedure many times, the stationary distribution of the resulting Markov chain on $G$ will be the desired conditional distribution under $\psi$ of $\Pi$ given $X$.

We can now combine EM with Metropolis--Hastings. At each step of the EM algorithm we could run the Metropolis--Hastings algorithm for a long time to get a good approximation of the conditional expectation (\ref{cond expec}). This naive approach is not very efficient. If $\psi^k$ does not differ much from $\psi^{k-1}$, then surely we should use the conditional expectation under $\psi^{k-1}$ to get the conditional expectation under $\psi^k$.

The stochastic approximation (SA) EM algorithm does just that. At each step of the EM algorithm, we update the current approximation of the conditional expectation by combining it with a single step of the Metropolis--Hastings algorithm, starting from the current value of $\Pi$ and using the current value of $\psi$. The weight of the added sample should depend on the iteration number. If at iteration $k$ we denote this weight by $\gamma(k)$, then we must have $\sum_k \gamma(k)=\infty$ and $\sum_k \gamma(k)^2 < \infty$, cf.\ \cite{delyon1999convergence}. We choose $\gamma(k)=1/(k+1000)$.

We can extend the algorithm in one important way. The ability to sample from $P_\psi(\Pi \mid X)$ means that once we have computed the MLE $\hat \psi$, we can also compute $\hat \theta_{\rm shuffle} =E_{\hat \psi}(\hat \psi_\Pi \mid X).$ We can actually fold this computation into the algorithm: At iteration $k$ we can update the current approximation of $\hat \theta_{\rm shuffle}$ by combining it with a single step of the Metropolis--Hastings algorithm, where the weight of the added sample is $\gamma(k)$.

\section{Examples}
\subsection{James--Stein}
Efron and Morris  \cite{efron1973stein} noted that the James--Stein estimator arises from an empirical Bayes procedure. Consider the following hierarchical model:
\begin{enumerate}
\item Sample $\theta_1,\theta_2,\dots,\theta_p$ i.i.d.\ $N(\mu,\sigma^2)$
\item Observe a sample $X$ from $\it{MVN}(\theta,I)$.
\end{enumerate}
Estimate $\mu$ and $\sigma$ and define the James--Stein estimator as
$$\hat \theta_{JS} = E_{\hat \mu , \hat \sigma} (\theta \mid X).$$

A striking demonstration of superiority of the James--Stein estimator over the usual estimator was given by Efron and Morris \cite{efron1975data}. They considered batting averages of  18 professional baseball players during the 1970 season. Their goal was to use the first 45 at-bats to predict the batting average during the remainder of the season. Note that the remainder of the season has many more than 45 at-bats. 

We follow \cite{efron1975data} and use an arcsine transformation to transform the data so that they are approximately normally distributed with unit variance. We compute the estimates of the transformed parameters, end then transform back. The total squared prediction error of the usual estimator is 0.086. The error of the James--Stein estimator is much lower at 0.027. We can compute the shuffled estimator with the SA-EM algorithm of the previous section. It is very similar to James--Stein, cf Figure \ref{fig:baseball18}, and its error is also 0.027

\begin{figure}[ht]
\begin{center}
\includegraphics[scale=0.8]{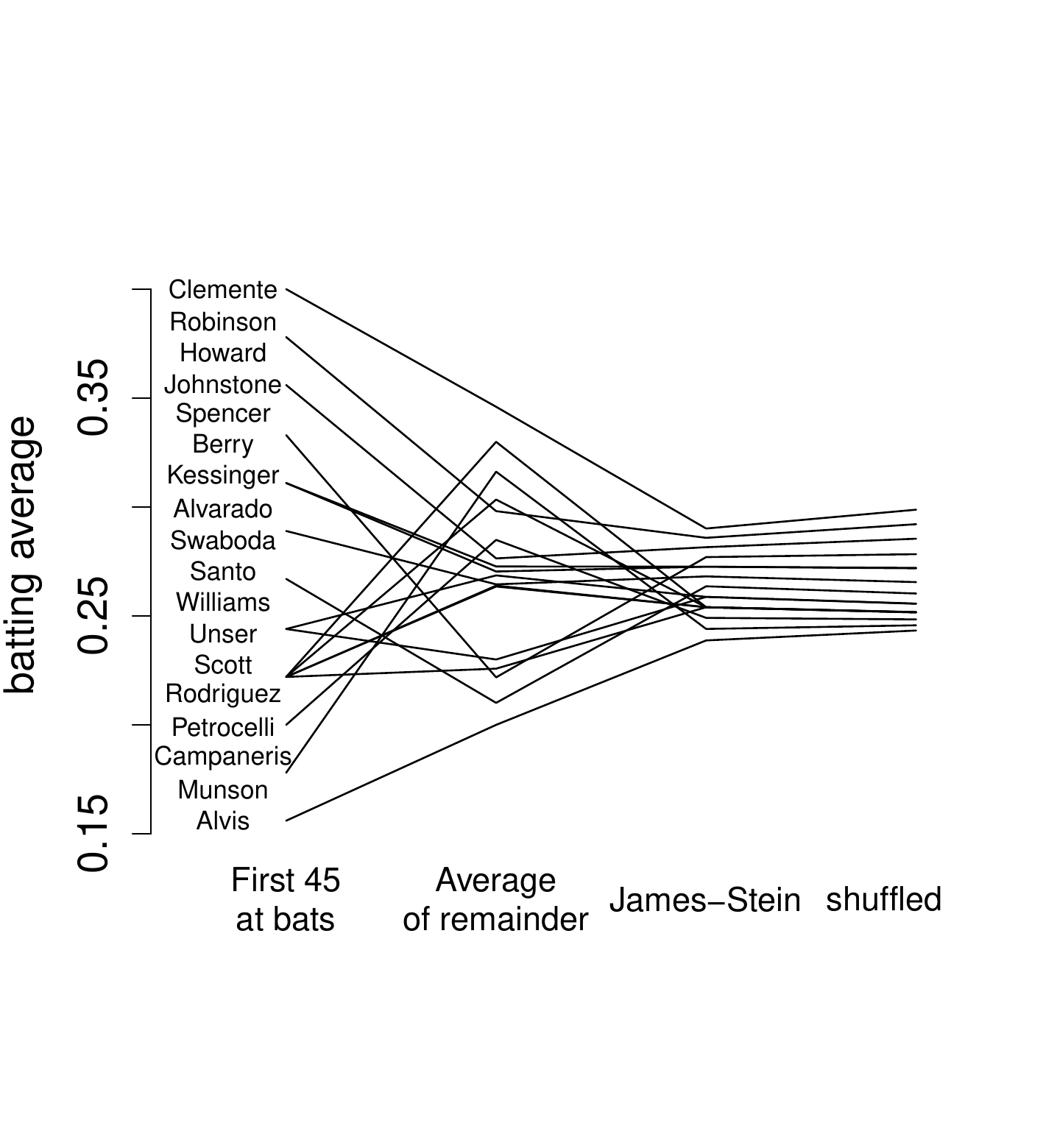}
\caption{Predicting batting averages.}\label{fig:baseball18}
\end{center}
\end{figure}

Efron and Morris \cite{efron1975data} state: ``We also wanted to include an unusually good hitter (Clemente) to test the method with at least one extreme parameter, a situation expected to be less favorable to Stein's estimator.'' To show that the shuffled estimator can be quite different from the James--Stein estimator, we change the data a little by making Clemente's performance even better than it already was. In reality, Clemente had 18 hits off his first 45 at-bats and 127 hits off the remaining 367 at bats. In Figure \ref{fig:baseball20} we show what the estimators would have been if Clemente's batting average had been about 4.4\% better, with 20 hits off his first 45 at-bats and 143 hits off the remaining 367 at bats.  The James--Stein estimator responds to the outlier by reducing the overall shrinkage factor. The shuffled estimator works differently. It tries to assign Clemente's performance to the other players (by shuffling the parameter), but that does not provide a very good fit to the data. Therefore,  Clemente is largely set apart from the other players. 

\begin{figure}[ht]
\begin{center}
\includegraphics[scale=0.8]{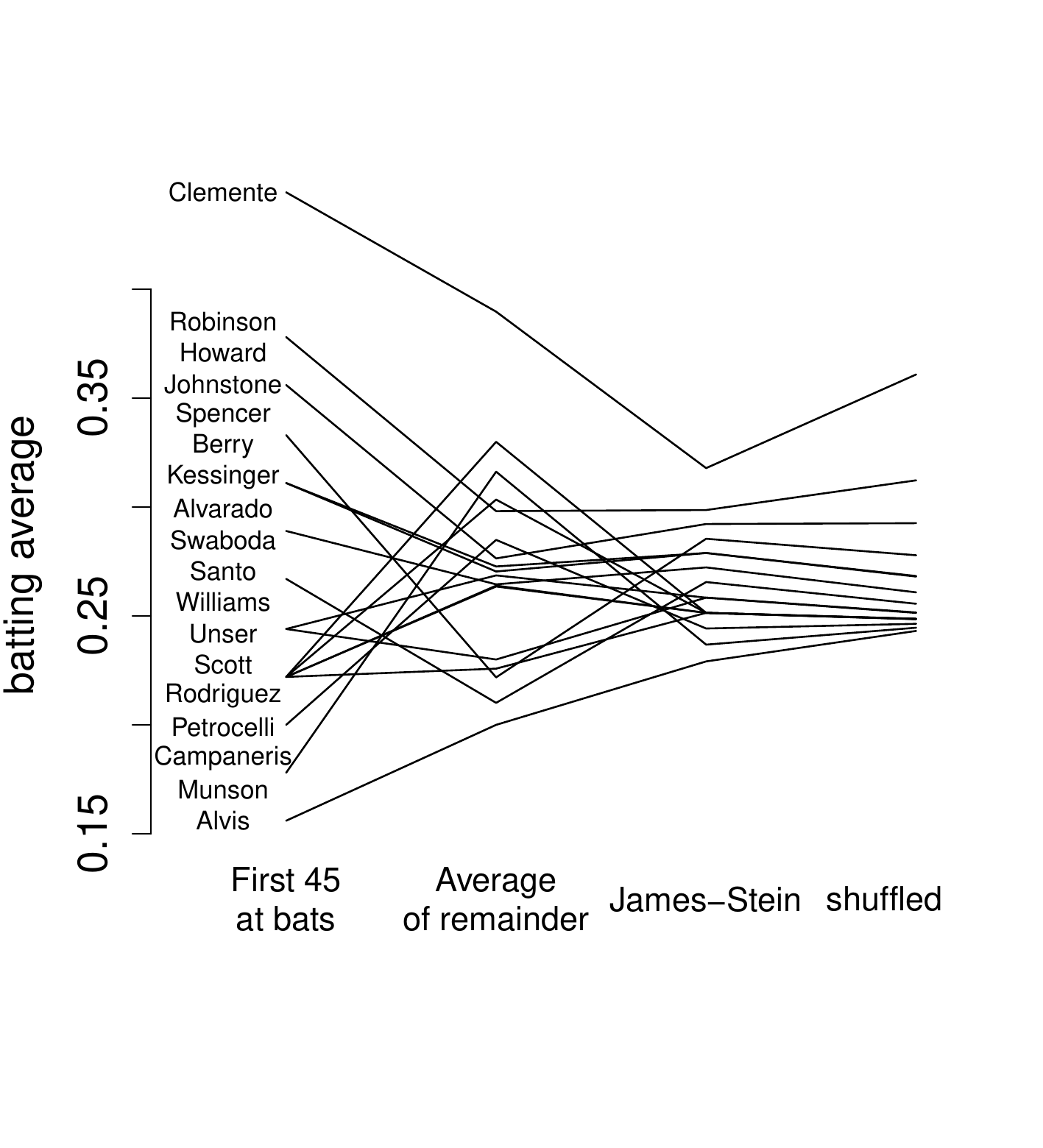}
\caption{Predicting batting averages.}\label{fig:baseball20}
\end{center}
\end{figure}

\subsection{Species sampling}
Suppose that we observe a sample $X=(X_1,X_2,\dots,X_p)$ from the multinomial distribution with parameters $n$ and $\theta$. The MLE of $\theta$ is just  the vector of empirical proportions. As we pointed out in the introduction, the MLE assigns no mass to the unobserved categories. If there are many probabilities that are small in relation to the sample size, then a substantial part of the probability mass will be unobserved.

There are various methods for estimating the probabilities of a multinomial distribution that do assign positive mass to the unobserved categories. For instance, a fully Bayesian approach with a Dirichlet prior for the vector $p$ would lead to adding pseudocounts to the observed counts. A limitation of this approach is that the estimates depend heavily on the parameters of the Dirichlet, see \cite{church1991comparison} for an extensive critique.

The approach of Good and Turing \cite{good1953population}, \cite{robbins1956empirical}, \cite{nadas1985turing} is perhaps the earliest example of non-parametric empirical Bayes. Consider the following hierarchical model.
\begin{enumerate}
\item Fix an arbitrary $\psi=(\psi_1,\psi_2,\dots,\psi_p)$ such that $\psi_i \geq 0$ and $\sum_{i=1}^p \psi_i=1$.
\item Sample $\theta$ uniformly at random from $\psi_1,\psi_2,\dots,\psi_p$.
\item Observe a sample $X$ from the \underline{bi}nomial distribution with parameters $n$ and $\theta$.
\end{enumerate}
Note that $\theta$ is now a (one-dimensional) random variable. The conditional mean of $\theta$ given $X = k$ is
\begin{equation}\label{post mean}
E(\theta \mid X = k) = \frac{\sum_{i=1}^p \psi_i^{k+1} (1-\psi_i)^{n-k}}{\sum_{i=1}^p \psi_i^{k}(1-\psi_i)^{n-k}}.
\end{equation}
Now,  let $N_{k,n}$ denote the number of categories that are observed $k$ times in a sample of size $n$. Good \cite{good1953population} notes that
\begin{equation}
E(\theta |X = k) = \frac{(k+1)E N_{k+1,n+1}}{(n+1)E N_{k,n}} \approx \frac{(k+1)E N_{k+1,n}}{(n+1)E N_{k,n}}\\
\end{equation}
and we can estimate $E(\theta |X = k)$ by replacing expectations with counts (and replacing $n+1$ by $n$)
\begin{equation}
\hat E(\theta |X = k)=\frac{(k+1)N_{k+1,n}}{n\,N_{k,n}}.
\end{equation}
We can use $\hat E(\theta |X = x_i)$ as an estimate of $\theta_i$  (probability of the $i$-th category). Moreover, we can estimate the total unobserved mass as
\begin{equation}\label{good turing}
\sum_{i : x_i=0} \hat E(\theta |X = x_i) = \frac{N_{1,n}}{n}
\end{equation}
which is just the proportion of ``singletons''.

The estimator of the unobserved mass performs very well in practice, but the estimates of the individual probabilities are not so good. For instance, they do not add up to one because probability mass is lost whenever some category is observed $k$ times, while no category is observed $k+1$ times. In fact, as $n$ grows large, it becomes increasingly unlikely that there is {\em any} pair of categories such that $x_j=x_i + 1$ and in the limit, the estimator will not assign any probability mass at all. Good was well aware of this, and he suggested smoothing the sequence $N_{1,n},N_{2,n},N_{3,n},\dots$.  Various smoothing techniques have been proposed since and we refer to \cite{chen1999empirical} for an overview and a comparative study. 

The shuffled model is both an extension and generalization of the model described above. Contrary to that model, the shuffled model fully describes the distribution of the data and hence it allows for maximum likelihood estimation of $\psi$ and, subsequently, estimation of $\theta$.

To demonstrate the shuffled model in the multinomial case, we have used the SA-EM algorithm to perform a modest simulation experiment. We sample a vector $\theta$ of length $p$ from the Dirichlet distribution with parameter vector $\alpha_1=\alpha_2=\dots,\alpha_{p}=1$. Next, we sample a vector $X$  from the multinomial distribution with parameters $n=50$ and $\theta$. We choose $p=5,10,25,50$, and then the expected unobserved mass is about 0.01, 0.03, 0.10, 0.25, respectively. We compare the total squared error of the MLE $\hat \theta = X/n$ and the shuffled estimator $\hat \theta_{\rm shuffle}$. As a benchmark, we  use the error of the posterior mean $\hat \theta_{\rm Bayes} = (X+1)/(n+p)$, which cannot be improved upon. We repeated the experiment 100 times, and show the results in Figure \ref{fig:species}. The shuffled estimator seems to perform slightly worse than the MLE when $p=5$ and 10, but decidedly better when $p=25$ and 50.

\begin{figure}[ht!]
\begin{center}
\includegraphics[scale=0.8]{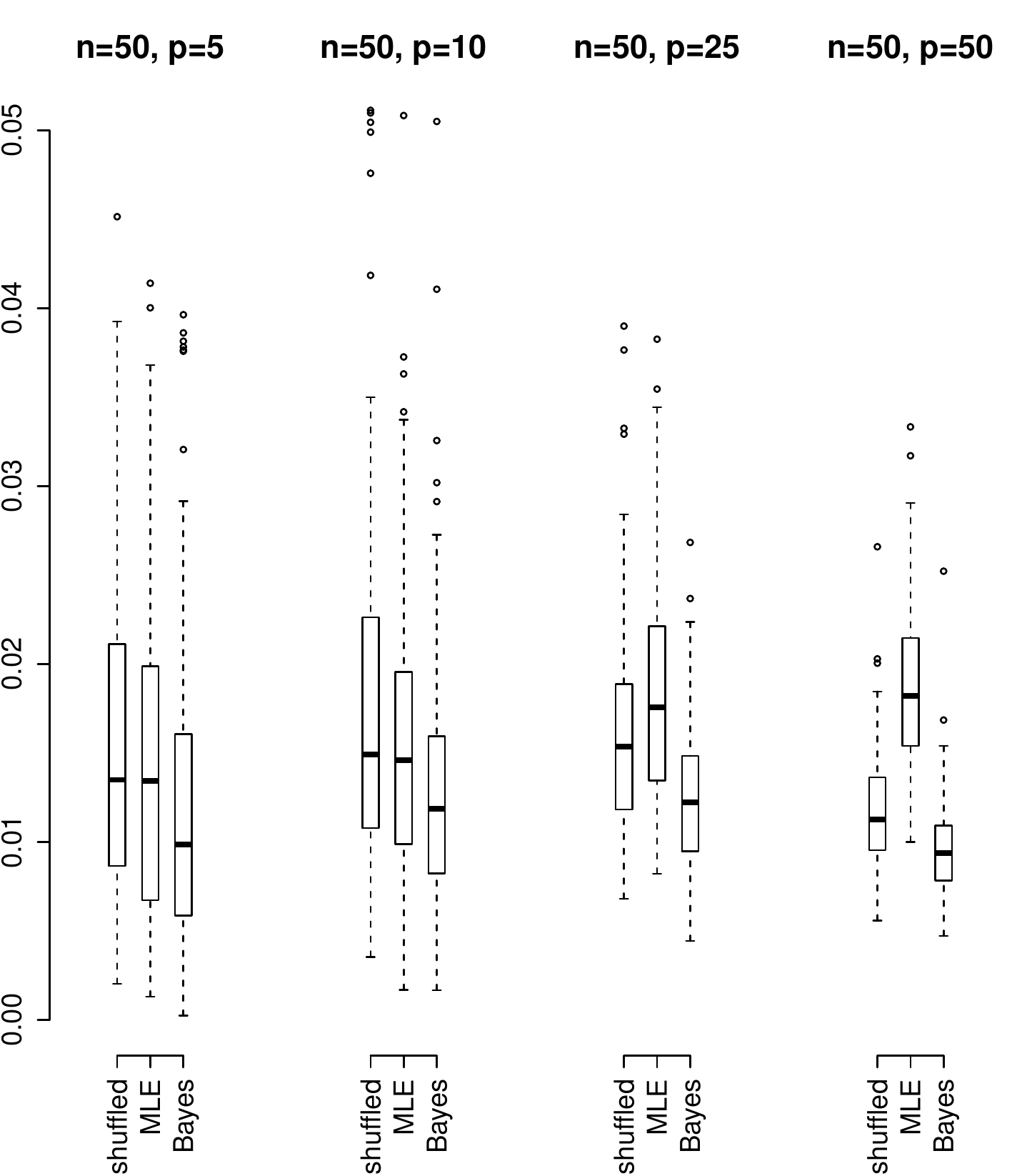}
\caption{Total squared error for estimating a multinomial distribution.}\label{fig:species}
\end{center}
\end{figure}

\section{Discussion}
We introduced {\em shuffling} as a method to assume partial exchangeability of a parameter vector $\theta$ with respect to a group $G$ of permutations. Shuffling can be viewed as an empirical Bayes approach, where we estimate the set of parameter values $\psi=\{\theta_1,\theta_2,\dots,\theta_p\}$ while putting a uniform prior their ordering. To be more precise, we should refer to $\psi$ as a multiset (with cardinality $p$) because the same parameter value could occur multiple times.

The shuffled model can also be motivated without assuming a partially exchangeable prior for $\theta$. We demonstrated that if we have a decision rule with a risk that is invariant  with respect to permutations in $G$, then the Bayes risk is {\it as if} $\theta$ has a partially exchangeable prior. 

The goal of shuffling is to ``borrow strength'' under minimal assumptions, and in this sense it is very similar to non-parametric empirical Bayes (NP-EB), cf.\  \cite{laird1978nonparametric}. While shuffling assumes exchangeability of the $\theta_i$, NP-EB assumes that the $\theta_i$ are i.i.d. Exchangeability is a weaker assumption than i.i.d. -- it is implied by it.  The difference may seem slight, but it does have several important consequences:

\begin{enumerate}
\item If $\theta$ represents the vector of probabilities of a multinomial model, then the $\theta_i$ can be exchangeable but not i.i.d.\ because they must add up to one.
\item The shuffled model is $p$ dimensional while the NP-EB model is infinite dimensional. This means that in the NP-EB model, the estimate of the prior can only converge to the true prior if $p$ (the number of units) becomes large. In the shuffled model, the estimate $\hat \psi$ can converge to $\psi$, if the information per unit increases.
\item The shuffled model concerns only the available units but the NP-EB model is generalizable to new units. Generalization to new units is appropriate if the available units are a random sample from all units.
\item The decision theoretic argument that provides an alternative motivation for the shuffled model does not seem to be available for the NP-EB model. 
\end{enumerate}

In this paper, we focused on parameter estimation. However, the shuffled model may be useful for any decision problem with permutation invariant risk, such as ranking exchangeable units (such as schools or hospitals) and testing multiple exchangeable hypotheses.

\section*{Acknowledgement}
I would like to thank Jelle Goeman for his help with understanding exchangeability, and Richard Gill for sharing his implementation of the SA-EM algorithm.

\section*{Appendix}
We provide {\tt R} code to compute the shuffled estimator in the baseball example.

\begin{verbatim}
saem = function(x,start,Niter){
  m=length(x)
  s2=1
  psi=start                                      # starting value for psi
  theta=start                      # starting value for E(theta | x, psi)
  perm=1:m                               # starting value for permutation
  for (iter in 1:Niter ){
    swap=sample(m,2); a=swap[1]; b=swap[2]                     # proposal
    mh=   (x[a]-psi[perm[a]])^2+(x[b]-psi[perm[b]])^2
    mh=mh-(x[a]-psi[perm[b]])^2-(x[b]-psi[perm[a]])^2
    mh=mh/(2*s2)
    if (runif(1) < exp(mh) | is.nan(mh)) {
      perm[c(a,b)]=perm[c(b,a)]                         # accept proposal
    }
    perm.inv=order(perm)                            # inverse permutation
    g=1/(1000+iter)
    psi=(1-g)*psi + g*x[perm.inv]
    theta=(1-g)*theta + g*psi[perm]
  }
  return(list(psi=psi,theta=theta))
}

ave45=c(18,17,16,15,14,14,13,12,11,11,10,10,10,10,10,9,8,7)/45
x=sqrt(45)*asin(2*ave45-1)                      # arcsine transformation
for (i in 1:10){                                       # do 10 re-starts
  fit=saem(x,start,100000)
  start=fit$psi
}
theta=(sin(fit$theta/sqrt(45))+1)/2                     # transform back
\end{verbatim}

\printbibliography

\end{document}